\documentclass{article}
\usepackage[utf8]{inputenc}
\usepackage{amsfonts}
\usepackage{graphicx}
\usepackage{geometry}
\renewcommand{\epsilon}{\varepsilon}
\renewcommand{\rho}{\varrho}
\usepackage{mathtools}
\usepackage{amssymb}
\usepackage{amsmath}
\usepackage{commath}
\usepackage{caption}
\usepackage{amsthm}
\usepackage{multirow}
\usepackage{float}
\usepackage{comment}
\usepackage{csquotes}
\usepackage{tikz}
\newtheorem{remark}{Remark}
\newtheorem{theorem}{Theorem}

\newtheorem{prop}{Proposition}
\title{Detection time of Dirac particles in one space dimension}
\author{A. Shadi Tahvildar-Zadeh and Stephanie Zhou\\
\small Department of Mathematics, Rutgers (New Brunswick)}
\date{May 2023}

\begin{document}

\maketitle
\begin{abstract}
     We consider particles emanating from a source point inside an interval in one-dimensional space and passing through detectors situated at the endpoints of the interval that register their arrival time. 
     Unambiguous measurements of arrival or detection time are problematic in the orthodox narratives of quantum mechanics, since time is not a self-adjoint operator. By contrast, the arrival time at the boundary of a particle whose motion is being guided by a wave function through the deBroglie-Bohm guiding law is well-defined and unambiguous, and can be computationally feasible provided the presence of detectors can be modeled in an effective way that does not depend on the details of their makeup.  
     We use an absorbing boundary condition for Dirac's equation (ABCD) proposed by Tumulka, which is meant to simulate the interaction of a particle initially inside a domain with the detectors situated on the boundary of the domain. By finding an explicit solution, we prove that the initial-boundary value problem for Dirac's equation satisfied by the wave function is globally well-posed, the solution inherits the regularity of the initial data, and depends continuously on it. 
     We then consider the case of a pair of particles emanating from the source inside the interval, and derive explicit formulas for the distribution of first arrival times at each detector, which we hope can be used to study issues related to non-locality in this setup. 
     \begin{center}
         Dedicated to the memory of Detlef D\"urr
     \end{center}
\end{abstract}

\section{Introduction}
In orthodox quantum mechanics, speaking of arrival/detection time of quantum particles is fraught with problems, since time does not lend itself to the self-adjoint operator formalism required for that approach\footnote{It is possible to define time more generally as a {\em positive operator-valued measure} \cite{Tum09}, but there is no unique way of doing that \cite{VHD13}.}.  Nevertheless, the arrival/detection time of a particle is something that  is routinely measured in time-of-flight (TOF) experiments performed in labs \cite{KM96,KPM97,PK97}.  There are numerous competing recipes in the literature for what the  distribution of arrival/detection times should be (see \cite{DD19,DND19,DS21,DasNoe21} for a critique of some of these approaches, and the possibility of using experiments to distinguish them.)  Many of these approaches ignore the presence of detection devices and consider the wave function of the particle to evolve unitarily, either under the free evolution or in presence of an external potential. As a result, they are prone to the ``backflow" problem, which can cause their proposed candidates for the arrival time probability density to become negative at points close to the source.   An alternative approach was taken by Tumulka \cite{TUM22,Tumulka,Tum22several}, based on an idea of Werner \cite{Wer87}, in which the presence of the detector is modeled through the imposition of an absorbing boundary condition on the Schr\"odinger flow (in the non-relativistic case) or the Dirac flow (in the relativistic case) of the wave function of the particle. Such a boundary condition ensures that particle velocities at the boundary are always outward. Tumulka showed that under such a boundary condition, his candidate for the probability density of the particle's arrival time is always non-negative.  Subsequently, Teufel and Tumulka succeeded in showing \cite{TeuTum19} that the corresponding boundary value problem for the wave function has a unique square-integrable solution, in both the relativistic and the non-relativistic cases.  Their existence proof uses techniques from functional analysis, and does not yield an explicit formula for the solution in either case.  

In this note we show that in one space dimension, the initial-boundary-value problem for the Dirac equation satisfied by the wave function of  a single spin-half particle, with absorbing boundary conditions corresponding to a pair of {\em ideal} detectors\footnote{See \cite{Tumulka} for the definition of ideal vs. non-ideal detectors.} placed at the two endpoints of an interval containing the particle source,  
is exactly solvable, and that the solution inherits the regularity of the initial data.  Applying Bohmian Mechanics rules then makes it possible for actual particle trajectories to be computed for any particle whose initial position is distributed randomly according to any given initial wave function, thereby setting the stage for comparisons to be made with other proposals for arrival time distribution, and the possibility of experimental testing of this theory.  

It is of interest to study how successful the ABCD method is in simulating actual detection, i.e. the interaction of the particle with the (presumably macroscopic) detecting apparatus that results in the device registering the presence of the particle at a particular time.  In order to avoid faster-than-light signaling, it must be the case that shifting one of the detectors by a small amount does not alter the distribution of arrival times at the other detector, or at least not before any possible light signal from the shifted detector has time to reach the other one. (Recall that in Bell-type experiments, changing the direction of the magnets affects only the correlations between the distributions, not the distributions themselves.) For a single particle, this is easily verified.  
We then derive explicit formulas for the distribution of arrival times at each detector, which in forthcoming work we plan to use to show that such superluminal signaling does not exist even when the initial wavefunction corresponds to a maximally entangled two-particle state.

\section{Absorbing boundary conditions for the Dirac equation}
Using the proposed equations in \cite{Tumulka}, we let $\Omega = (-L,L)$ be an open interval in $\mathbb{R}$ and let $\psi = \begin{pmatrix} \psi_+\\ \psi_-\end{pmatrix}:[0,\infty)\times\Omega\rightarrow\mathbb{C}^{2}$ be the unique solution of the initial-boundary value problem (IBVP)
\begin{equation}\label{rodi_eq}
\begin{cases}
\begin{array}{rcll}
     ic\hbar\gamma^{\mu}\partial_{\mu}\psi & = & mc^{2}\psi & \\
     \psi(0,s) & = & \mathring{\psi}(s); & s\in\Omega,\qquad \mathring{\psi}\in C^k_{c}(\Omega),\ k\geq 0 \\
     \mathbf{n}(s)\cdot\mathbf{\alpha}\ \psi(t,s) & = & \psi(t,s); & t\geq 0,s\in\partial\Omega.
\end{array}
\end{cases}
\end{equation}
Here $\{\gamma^0,\gamma^1\}$ are Dirac gamma matrices, $m$ is the rest mass of the spin-1/2 particle, $c$ is the speed of light in vacuum, $\hbar$ is Planck's constant,  $\mathbf{n}$ is the normal to $\partial \Omega$, and $\mathbf{\alpha} = \alpha^1 := \gamma^0\gamma^1$.

The initial data $\mathring{\psi}$ is assumed to be $C^k$, for a fixed integer $k\geq 0$, and compactly supported inside the interval $\Omega$.  The data can therefore be extended outside $\Omega$ to be identically zero. In the following, when speaking of $\mathring{\psi}$ we always have this extension in mind.

The boundary of the spacetime domain for $\psi$ is the set of points $\{(t,-L),(t,L)\}$. So we have $\mathbf{n}=1$ at $(t,L)$, $\mathbf{n}=-1$ at $(t,-L)$.  Choosing $\gamma^0 = \begin{pmatrix} 0 & 1\\ 1 & 0\end{pmatrix}$, $\gamma^1 = \begin{pmatrix} 0 & -1\\ 1 & 0\end{pmatrix}$, we have
 $\mathbf{\alpha}=\begin{pmatrix}1 & 0 \\ 0 & -1\end{pmatrix}$.

Plugging these into the boundary condition in \eqref{rodi_eq}, which we call an {\em Absorbing Boundary Condition for the Dirac equation} (ABCD), we get $\psi_{+}(t,L)=0$ and $\psi_{-}(t,-L)=0$ for all $t\geq 0$. Now plugging these two boundary conditions into the Dirac equation, we get the additional boundary condition $m\psi_{\mp}(t,\pm L)=\mp i\partial_{s}\psi_{\pm}(t,\pm L)$.
This gives us an equivalent IBVP for the Klein-Gordon equation obtained by iterating the Dirac operator in \eqref{rodi_eq}:
\begin{equation}\label{ibvp_arrival}
\begin{cases}
\begin{array}{rcl}
     \partial_{t}^{2}\psi_{\pm}-\partial_{s}^{2}\psi_{\pm}+m^{2}\psi_{\pm} & = & 0 \\
     \psi_{\pm}(0,s) & = & \mathring{\psi}_{\pm}(s) \in C^k_c(\Omega), \ k\geq 0 \\
     \partial_{t}\psi_{\pm}(0,s) & = & -im\mathring{\psi}_{\mp}(s)\pm\partial_{s}\mathring{\psi}_{\pm}(s) \\
     \psi_{\pm}(t,\pm L)& = & 0 \\
     m\psi_{\mp}(t,\pm L)\pm i\partial_{s}\psi_{\pm}(t,\pm L) & = & 0
\end{array}
\end{cases}
\end{equation}
(We have set $c=\hbar=1$.)\newline
The following theorem shows that \eqref{ibvp_arrival} has an explicit solution by splitting the IBVP into its initial and boundary value problem parts and finding the solutions to each. Although this solution is in the form of an infinite series, we shall see that at any fixed time there are only finitely many nonzero terms in it, so that there are no convergence issues.

\begin{theorem}
The IBVP in \eqref{ibvp_arrival} has a unique solution that is as regular as its initial data and depends continuously on it.
\end{theorem}

\begin{proof}
In order to solve these equations we  set $\psi_{\pm}=\Phi_{\pm}+\chi_{\pm}$, where $\Phi_{\pm}$, defined on $[0,\infty)\times (-\infty,\infty)$, are the solutions to
\begin{equation}\label{kg}
\begin{cases}
\begin{array}{rcl}
     \partial_{t}^{2}\Phi_{\pm}-\partial_{s}^{2}\Phi_{\pm}+m^{2}\Phi_{\pm} & = & 0 \\
     \Phi_{\pm}(0,s) & = & \mathring{\psi}_{\pm}(s) \\
     \partial_{t}\Phi_{\pm}(0,s) & = & -im\mathring{\psi}_{\mp}(s)\pm\partial_{s}\mathring{\psi}_{\pm}(s),
\end{array}
\end{cases}
\end{equation}
and where $\chi_{\pm}$ are functions defined on $[0,\infty)\times (-\infty,L)$, resp. $[0,\infty)\times (-L,\infty)$ that satisfy
\begin{equation}\label{chipm}
\begin{cases}
\begin{array}{rcl}
      \partial_{t}^{2}\chi_{\pm}-\partial_{s}^{2}\chi_{\pm}+m^{2}\chi_{\pm} & = & 0 \\
      \chi_{\pm}(0,s) & = & 0 \\ 
      \partial_{t}\chi_{\pm}(0,s) & = & 0 \\ 
      \chi_{\pm}(t,\pm L) & = & f_{\pm}(t) \\
      m\chi_{\mp}(t,\pm L)\pm i\partial_{s}\chi_{\pm}(t,\pm L) & = & i \partial_t f_\pm(t),
\end{array}
\end{cases}
\end{equation}
where $f_{\pm}(t):=-\Phi_{\pm}(t,\pm L)$.\newline
\eqref{kg} is the Cauchy problem for the one-dimensional Klein-Gordon equation, in older literature called the (lossless) {\em telegraph} equation (e.g. \cite{CHII}, p.544), whose solution is well-known (see e.g. \cite{Holden}):
\begin{equation}\label{sol_phi}
    \begin{split}
    \Phi_{\pm}(t,s) & =\frac{1}{2}[\mathring{\psi}_{\pm}(s-t)+\mathring{\psi}_{\pm}(s+t)]    -\frac{tm}{2}\int_{s-t}^{s+t}\frac{J_{1}(m\sqrt{t^{2}-(s-\sigma)^{2}})}{\sqrt{t^{2}-(s-\sigma)^{2}}}\mathring{\psi_{\pm}}(\sigma)d\sigma \\ & -\frac{1}{2}\int_{s-t}^{s+t}J_{0}(m\sqrt{t^{2}-(s\newline -\sigma)^{2}})(im\mathring{\psi}_{\mp}(\sigma)\mp\partial_{\sigma}\mathring{\psi}_{\pm}(\sigma))d\sigma,
    \end{split}
\end{equation}
where $J_n$ is the Bessel function of order $n$.
Using integration by parts, the above can be rewritten as follows:
\begin{equation}\label{sol_phi_byparts}
    \Phi_{\pm}(t,s) =\mathring{\psi}_{\pm}(s\pm t) -\frac{m}{2}\int_{s-t}^{s+t} Z_1(t,\pm(s-\sigma))
   \mathring{\psi_{\pm}}(\sigma)d\sigma  -\frac{im}{2}\int_{s-t}^{s+t}Z_{0}(t,s-\sigma)\mathring{\psi}_{\mp}(\sigma)d\sigma,
\end{equation}
where we define, for $\nu\geq 0$,
\begin{equation}
    Z_\nu(t,s) :=  J_{\nu}(m\sqrt{t^{2}-s^{2}})\left(\frac{t- s}{t+ s}\right)^{\nu/2}.
\end{equation}
On the other hand, we can solve \eqref{chipm} using Laplace transform methods, as follows:
We first find the general solution to the first three equations of \eqref{chipm}:
\begin{equation}\label{gensol}
\begin{array}{rcl} 
    \Tilde{\chi}_{+}(p,s) & = & c_{1}(p)e^{ks}+c_{2}(p)e^{-ks} \\
    \Tilde{\chi}_{-}(p,s) & = & c_{3}(p)e^{ks}+c_{4}(p)e^{-ks},
\end{array}
\end{equation} 
where
\begin{equation}
     k := \sqrt{p^{2}+m^{2}},\qquad \mbox{Re}(p)>0,
\end{equation}
and tilde denotes the Laplace transform, i.e. $\tilde{\Psi}(p) = \int_0^\infty f(t)e^{-p t} dt$.
Applying this to the boundary conditions in \eqref{chipm}, we obtain
\begin{equation}\label{gensystem}
\begin{array}{rcl} 
    \tilde{f}_{+}(p) & = & c_{1}(p)e^{kL}+c_{2}(p)e^{-kL} \\
    \tilde{f}_{-}(p) & = & c_{3}(p)e^{-kL}+c_{4}(p)e^{kL} \\
    ip\tilde{f}_{+}(p) & = & m(c_{3}(p)e^{kL}+c_{4}(p)e^{-kL})+i(kc_{1}(p)e^{kL}-kc_{2}(p)e^{-kL}) \\
    ip\tilde{f}_{-}(p) & = & m(c_{1}(p)e^{-kL}+c_{2}(p)e^{kL})+i(-kc_{3}(p)e^{-kL}+kc_{4}(p)e^{kL}).
\end{array}
\end{equation} 

After solving the system \eqref{gensystem} for $c_{1},c_{2},c_{3},c_{4}$ and plugging back into \eqref{gensol}, we arrive at
{\small
\begin{eqnarray*}
       \Tilde{\chi}_{+}(p,s) & = & \frac{\Tilde{f}_{+}(m^{2}e^{6kL}+2k^{2}e^{2kL}-m^{2}e^{2kL}+2kpe^{2kL})+i\Tilde{f}_{-}(kme^{4kL}-mpe^{4kL}+km+mp)}{D}e^{k(L+s)} \\ 
& + &\frac{\Tilde{f}_{+}(2k^{2}e^{4kL}-2kpe^{4kL}-m^{2}e^{4kL}+m^{2})+i\Tilde{f}_{-}(mpe^{6kL}-kme^{6kL}-mpe^{2kL}-kme^{2kL})}{D}e^{k(L-s)}
\\
       \Tilde{\chi}_{-}(p,s) & = & \frac{i\Tilde{f}_{+}(mpe^{6kL}-kme^{6kL}-mpe^{2kL}-kme^{2kL})+\Tilde{f}_{-}(2k^{2}e^{4kL}-2kpe^{4kL}-m^{2}e^{4kL}+m^{2})}{D}e^{k(L+s)} \\ 
& + & \frac{i\Tilde{f}_{+}(kme^{4kL}-mpe^{4kL}+km+mp)+\Tilde{f}_{-}(m^{2}e^{6kL}+2k^{2}e^{2kL}-m^{2}e^{2kL}+2kpe^{2kL})}{D}e^{k(L-s)},
\end{eqnarray*} 
}
where $D:=m^2(e^{2kL}+i\frac{k+p}{m})(e^{2kL}-i\frac{k+p}{m})(e^{2kL}+i\frac{k-p}{m})(e^{2kL}-i\frac{k-p}{m})$.

After doing a partial fraction decomposition, 
the above can be rewritten as

\begin{equation}
  \label{post-decomposition}  
    \Tilde{\chi}_{\pm}(p,s) = \frac{1}{2}
    \bigg[
    e^{-k(L\mp s)}\Big(\pm\frac{\Tilde{f}_{+}(p)-\Tilde{f}_{-}(p)}{1+i\frac{k-p}{m}e^{-2kL}}+\frac{\Tilde{f}_{+}(p)+\Tilde{f}_{-}(p)}{1-i\frac{k-p}{m}e^{-2kL}}\Big)+ i\frac{k-p}{m}e^{-k(L\pm s)}\Big(\pm\frac{\Tilde{f}_{+}(p)-\Tilde{f}_{-}(p)}{1+i\frac{k-p}{m}e^{-2kL}}- \frac{\Tilde{f}_{+}(p)+\Tilde{f}_{-}(p)}{1-i\frac{k-p}{m}e^{-2kL}}\Big)
    \bigg].
 \end{equation}   

Noting that $\left|i\frac{k-p}{m}e^{-2kL}\right|<1$, we can view \eqref{post-decomposition} as the sum of four convergent geometric series, so that

\begin{equation}\label{series}
\begin{array}{rcl} 
  \begin{aligned}
      \Tilde{\chi}_{\pm}(p,s) & = 
     \frac{1}{2}
     \bigg[
     \pm\sum_{n=0}^{\infty}\Big(-i\frac{k-p}{m}\Big)^{n}e^{-k[(2n+1)L\mp s]}(\Tilde{f}_{+}(p)-\Tilde{f}_{-}(p)) 
     +\sum_{n=0}^{\infty}\Big(i\frac{k-p}{m}\Big)^{n}e^{-k[(2n+1)L\mp s]}(\Tilde{f}_{+}(p)+\Tilde{f}_{-}(p)) \\ &
       \mp\sum_{n=0}^{\infty}\Big(-i\frac{k-p}{m}\Big)^{n+1}e^{-k[(2n+1)L\pm s]}(\Tilde{f}_{+}(p)-\Tilde{f}_{-}(p))  
       -\sum_{n=0}^{\infty}\Big(i\frac{k-p}{m}\Big)^{n+1}e^{-k[(2n+1)L\pm s]}(\Tilde{f}_{+}(p)+\Tilde{f}_{-}(p)) 
       \bigg] 
   \end{aligned}
\end{array}
\end{equation} 

From a table of inverse Laplace transforms (e.g. \cite{Oberhettinger}, formula 14.52) we find
$$
     \mathcal{L}^{-1}[(\frac{(k-p)^{\nu}}{k})e^{-kx}](\tau,x) = m^{\nu} (\frac{\tau-x}{\tau+x})^{\frac{1}{2}\nu}J_{\nu}(m\sqrt{\tau^{2}-x^{2}})H(\tau-x) = m^\nu Z_\nu(\tau,x)H(\tau - x)
$$
where $H$ is the Heaviside function $H(t) =1$ for $t>0$ and $0$ otherwise.\newline
We therefore have the following explicit solution for $\chi_{\pm}$:

\begin{equation}\label{transform}
\begin{array}{rcl} 
    \begin{aligned}
        \chi_{\pm}(t,s) & = 
        \frac{1}{2}\sum_{n=0}^{\infty} i^{n}
        \Bigg[
        H(t-(2n+1)L\pm s)\bigg((-1)^{n}      \frac{d}{ds}  \int_{(2n+1)L\mp s}^{t} d\xi \ F_{-}(t-\xi) 
        Z_{n}(\xi,\pm((2n+1)L-s))  \\ & 
        \hspace{2.1in}\pm      \frac{d}{ds}  \int_{(2n+1)L\mp s}^{t} d\xi \ F_{+}(t-\xi)  Z_{n}(\xi,\pm((2n+1)L-s))\bigg)  \\ & 
        +iH(t-(2n+1)L\mp s)\bigg((-1)^{n+1}    \frac{d}{ds}   \int_{(2n+1)L\pm s}^{t} d\xi \ F_{-}(t-\xi)  
        Z_{n+1}(\xi,\pm((2n+1)L-s)) \\ & 
        \hspace{1.7in}\pm    \frac{d}{ds}   \int_{(2n+1)L\pm s}^{t} d\xi \ F_{+}(t-\xi)  Z_{n+1}(\xi,\pm((2n+1)L-s))\bigg),
        \Bigg]
    \end{aligned}
\end{array}
\end{equation} 
where
$$
F_\pm := f_+ \pm f_-.
$$
Carrying out the differentiations in \eqref{transform}, we have
\begin{equation}\label{diff}
    \begin{split}
         \chi_{\sigma}(t,s) & = H(t-L+\sigma s)f_{\sigma}(t-L+\sigma s) \\ & 
         \hspace{.2in}+\sum_{n=0}^{\infty} i^{n}\Bigg(H(t-(2n+1)L+\sigma s)\int_{(2n+1)L-\sigma s}^{t}f_{(-1)^{n}\sigma}(t-\xi)R_{n}(\xi,(2n+1)L-\sigma s) d\xi \\ & 
         \hspace{.45in}-iH(t-(2n+1)L-\sigma s)\int_{(2n+1)L+\sigma s}^{t}f_{(-1)^{n+1}\sigma}(t-\xi)R_{n+1}(\xi,(2n+1)L+\sigma s)d\xi\Bigg),
    \end{split}
\end{equation}
where $\sigma\in\{+,-\}$, and for $k\geq 0$,
$$
    R_{k}(\xi,\eta) := -m\eta (\xi-\eta)^{k}\frac{J_{k+1}(m\sqrt{\xi^{2}-\eta^{2}})}{(\xi^{2}-\eta^{2})^{\frac{k+1}{2}}}
    +k(\xi-\eta)^{k-1}\frac{J_{k}(m\sqrt{\xi^{2}-\eta^{2}})}{(\xi^{2}-\eta^{2})^{\frac{k}{2}}}.
$$
Adding $\chi_{\pm}$ to $\Phi_{\pm}$, we arrive at the solution for $\psi_{\pm}$.
This solution can clearly be put in terms of the initial data, as shown for example by the $n=0$ term of $\chi_\sigma$, which we denote by ${\chi_\sigma}_{0}$:
{\small
\begin{equation}\label{chi0}
    \begin{split}
        {\chi_{\sigma}}_{0}(t,s) & = -H(t-L+\sigma s)\Bigg[\Phi_{\sigma}(t-L+\sigma s,\sigma L)+\int_{L-\sigma s}^{t}R_{0}(\xi,L-\sigma s)\Bigg(\mathring{\psi}_{\sigma}(\sigma(L+t-\xi)) \\ &
        \hspace{.2in}-\frac{m}{2}\int_{\sigma L-t+\xi}^{\sigma L+t-\xi}Z_1(t-\xi,\sigma L - \varsigma)\mathring{\psi}_{\sigma}(\varsigma) d\varsigma 
        -\frac{im}{2}\int_{\sigma L-t+\xi}^{\sigma L+t-\xi}Z_{0}(t-\xi,\sigma L-\varsigma)\mathring{\psi}_{-\sigma}(\varsigma)d\varsigma \Bigg) d\xi\Bigg]  \\ &
        +iH(t-L-\sigma s)\int_{L+\sigma s}^{t}R_{1}(\xi,L+\sigma s)\Bigg(\mathring{\psi}_{-\sigma}(-\sigma(L+t-\xi)) \\ &
        \hspace{.2in}-\frac{m}{2}\int_{-\sigma L-t+\xi}^{-\sigma L+t-\xi}Z_1(t-\xi,-\sigma L-\varsigma)\mathring{\psi}_{-\sigma}(\varsigma) d\varsigma 
        -\frac{im}{2}\int_{-\sigma L-t+\xi}^{-\sigma L+t-\xi}Z_{0}(t-\xi,\sigma L+\varsigma)\mathring{\psi}_{\sigma}(\varsigma)d\varsigma \Bigg)d\xi
    \end{split}
\end{equation}
}
Furthermore, the Heaviside functions appearing in \eqref{diff} show that $\chi_{n}$, the $n$-th term in the summation, is supported in $\bigcup_{k=0}^{n}\mathcal{R}_{k}$, with regions $\mathcal{R}_{k}$ shown in Fig.~\ref{regions}.  This shows that for any fixed $t>0$ there are only finitely many non-zero terms in \eqref{diff}, so convergence is never an issue. 

Finally, the behavior of Bessel functions for small values of their argument, $J_n(x) = O(x^n)$ as $x\to 0$, imply that, despite appearances, the solution kernels $Z_k$ and $R_k$ are smooth and bounded on compact domains, which ensures that the solution is as regular as the data and depends continuously on it.
\end{proof}
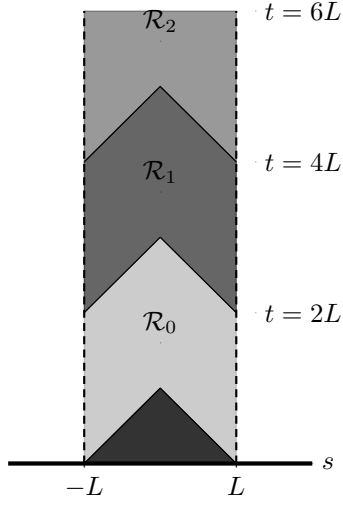
\begin{figure}
    \centering
    \begin{tikzpicture}[scale=0.50]
    \draw[black, ultra thick] (-4,0) -- (4,0);
    \draw[black, thin] (-2,-.1) -- (-2,.1);
    \draw[black, thin] (2,-.1) -- (2,.1);
    \draw[black, thick, densely dashed] (-2,0) -- (-2,12);
    \draw[black, thick, densely dashed] (2,0) -- (2,12);
    
    \draw[black, thin] (-2,0) -- (0,2);
    \draw[black, thin] (0,2) -- (2,0);
    \draw[black, thin] (-2,4) -- (0,6);
    \draw[black, thin] (0,6) -- (2,4);
    \draw[black, thin] (-2,8) -- (0,10);
    \draw[black, thin] (0,10) -- (2,8);
    
    \draw[fill=black, opacity=.8] (-2,0) -- (0,2) -- (2,0) -- cycle;
    \draw[fill=black, opacity=.2] (-2,0) -- (-2,4) -- (0,6) -- (2,4) -- (2,0) -- (0,2) -- cycle;
    \draw[fill=black, opacity=.6] (-2,4) -- (-2,8) -- (0,10) -- (2,8) -- (2,4) -- (0,6) -- cycle;
    \draw[fill=black, opacity=.4] (-2,8) -- (-2,12) -- (2,12) -- (2,8) -- (0,10) -- cycle;
    
    \filldraw[black] (4,0) circle (0pt) node[anchor=west] {$s$};
    \filldraw[black] (-2,-.1) circle (0pt) node[anchor=north] {$-L$};
    \filldraw[black] (2,-.1) circle (0pt) node[anchor=north] {$L$};
    \filldraw[black] (2.5,4) circle (0pt) node[anchor=west] {$t=2L$};
    \filldraw[black] (2.5,8) circle (0pt) node[anchor=west] {$t=4L$};
    \filldraw[black] (2.5,12) circle (0pt) node[anchor=west] {$t=6L$};
    
    \filldraw[black] (0,3.2) circle (0pt) node[anchor=south] {$\mathcal{R}_{0}$};
    \filldraw[black] (0,7.2) circle (0pt) node[anchor=south] {$\mathcal{R}_{1}$};
    \filldraw[black] (0,11.2) circle (0pt) node[anchor=south] {$\mathcal{R}_{2}$};
    
    \end{tikzpicture}
    \caption{Support of $\chi_n$}
    \label{regions}
\end{figure}
\section{Arrival Time of Bohmian Trajectories}

Let $\psi = \begin{pmatrix} \psi_+\\ \psi_-\end{pmatrix}$ be the wave function of a spin-1/2 particle in $\mathbb{R}^{1,1}$.  The {\em Dirac current} 
\begin{equation}\label{def:j}
    j^\mu := \overline{\psi} \gamma^\mu \psi,\qquad \overline{\psi} := \psi^\dagger \gamma^0
\end{equation}
is the simplest Lorentz vector that can be constructed from the Dirac bispinor $\psi$.  When $\psi$ satisfies the Dirac equation $-i\gamma^\mu \partial_\mu \psi + m \psi = 0$, it follows that the vectorfield $j$ is divergence free, i.e. 
\begin{equation}\label{divj}
    \partial_\mu j^\mu  = 0.
\end{equation}
Setting $\rho(t,s) := j^0 = \psi^\dagger \psi$ and $J(t,s) := j^1 = \psi^\dagger \alpha^1 \psi$,  the above can be written as $\partial_t \rho + \partial_s J = 0$, which has the form of an equation of continuity for a density $\rho$. The quantity $v(t,s) = J/\rho$ is thus naturally a velocity field defined on the 1-particle configuration space.  

Let $Q(t)$ denote the actual position of the particle at time $t$. According to the principles of Bohmian Mechanics (e.g. \cite{DGZ13}, Chap. 9) the guiding equation for the motion of the particle is
\begin{equation}\label{guiding}
    \frac{dQ(t)}{dt}=\frac{J(t,Q(t))}{\rho(t,Q(t))}.
\end{equation}
The above ODE can be uniquely solved given initial data $Q_0 = Q(0) \in \mathbb{R}$ which is assumed to be distributed randomly according to the initial density $\rho_0 := \mathring{\psi}^\dagger \mathring{\psi}$.  The {\em arrival time} of the trajectory $s= Q(t)$ at the boundary $\partial\Omega$ of the domain is simply 
\begin{equation}
    T = \inf\{t>0\ |\ Q(t) \in \partial\Omega\}
\end{equation}
(Recall that the infimum of the empty set is by definition $+\infty$.)





Let $j^{\mu}$ be as in \eqref{def:j}. Let $Z=(T,\mathbf{X})$ where  $\mathbf{X}$ is the location on $\partial\Omega$ where a particle gets first detected at time $T$, meaning there is a trajectory $Q(t)$ with $Q(0)\in \Omega$ that at time $t=T$ reaches the boundary for the first time, and $Q(T) = \mathbf{X}$. 
(If $T = \infty$, we write $Z=\infty$.) $Z$ is a random variable defined on $\Omega$, and according to \cite{Tumulka} the  distribution $\mu$ of $Z$ satisfies
    \begin{equation}\label{mu}
     \mu(t_{1}\leq T<t_{2},\boldsymbol{X}\in\partial\Omega) = \int_{t_{1}}^{t_{2}}dt\int_{\partial\Omega}\boldsymbol{n}(\boldsymbol{x})\cdot\boldsymbol{j}^{\psi_{t}}(\boldsymbol{x})d\sigma = \int_{t_{1}}^{t_{2}}\left(j^{1}(t,L)-j^{1}(t,-L)\right)dt.
     \end{equation}
It was shown in \cite{Tumulka} that
when the absorbing boundary conditions are satisfied, the integrand 
in \eqref{mu} is a probability density function on $\mathbb{R}^+$. Moreover, since in this case the boundary is made up of just two points, call them $A$ for Alice and $B$ for Bob, the distribution density of arrival times at Alice's detector on the left and for Bob's detector on the right are
\begin{equation}
    \rho_A(t) = -j^1(t,-L) =  |\psi_+(t,-L)|^2,\qquad \rho_B(t) = j^1(t,L) = |\psi_-(t,L)|^2.
\end{equation}
\begin{remark}
We note that these are not normalized probability density functions, since, as was shown in \cite{Tumulka}, a simple application of the Divergence Theorem implies that

\begin{equation}\label{idk}
         \int_{-L}^{L}|\mathring{\psi}(s)|^{2}ds-\int_{-L}^{L}|\psi(t,s)|^{2}ds = \int_{0}^{t}\left(j^{1}(T,L)-j^{1}(T,-L)\right)dT = \int_{0}^{t}\rho_B(T)+\rho_A(T)\ dT.
\end{equation}
Moreover,  since $\mathring{\psi}$ is assumed to be compactly supported in $[-L,L]$, by the Born Rule the left hand side of \eqref{idk} is $$1-\mbox{Prob}(T=t,X\in\Omega).$$
 Therefore, letting $t \to \infty$ we can only conclude that the {\em sum} of the total integrals of $\rho_A$ and $\rho_B$ is {\em at most} equal to one.
 \end{remark}

\section{Two-body problem}
Let us now imagine that the particle source in the middle of the interval emits {\em pairs} of particles.  The two body multi-time wave function $\psi = (\psi_{--},\psi_{-+},\psi_{+-},\psi_{++})$ for this problem has four components, each one of which is a function of two time and two space variables:  $\psi_{\sigma_1\sigma_2} = \psi_{\sigma_1\sigma_2}(x_1,x_2)$ with $x_i^0 = t_i, x_i^1 = s_i$, $i=1,2$.  It satisfies the multi-time system of Dirac equations
\begin{equation}\label{multiDir}
 i\gamma_j^{\mu}\partial_{x_j^\mu}\psi  =  m\psi   \qquad j = 1,2
\end{equation}
where $\gamma_1^\mu := \gamma^\mu \otimes \mathbf{1}$ and $\gamma_2^\mu := \mathbf{1}\otimes \gamma^\mu$.  
The above equations are amended by the boundary conditions
\begin{equation}\label{2boundary}
    \psi_{\pm\sigma_2}(t,\pm L,t,s) = 0,\qquad \psi_{\sigma_1\pm}(t,s,t,\pm L) = 0,\qquad \forall \sigma_1,\sigma_2\in \{+,-\}, \forall t \geq 0, \forall s \in[-L,L]
\end{equation}
The two-body tensor current $j$ is defined as
$
        j^{\mu\nu} = \overline{\psi}(\gamma^{\mu}\otimes\gamma^{\nu})\psi
$
where $\overline{\psi} := \psi^\dagger \gamma^0\otimes \gamma^0$.  It follows that
\begin{eqnarray}
    j^{00} & = & |\psi_{--}|^2 + |\psi_{-+}|^2 + |\psi_{+-}|^2 + |\psi_{++}|^2 \\
    j^{01} & = & |\psi_{--}|^2 - |\psi_{-+}|^2 + |\psi_{+-}|^2 - |\psi_{++}|^2 \\
    j^{10} & = & |\psi_{--}|^2 + |\psi_{-+}|^2 - |\psi_{+-}|^2 - |\psi_{++}|^2 
\end{eqnarray}
When $\psi$ satisfies the multi-time Dirac system \eqref{multiDir} the tensor current $j^{\mu\nu}$ satisfies a pair of conservation laws:
\begin{equation}\label{2cons}
    \partial_{t_1} j^{00} + \partial_{s_1} j^{10} = 0,\qquad \partial_{t_2} j^{00} + \partial_{s_2} j^{01} = 0.
\end{equation}
Suppose that there is a particle detector at each endpoint of the interval $\Omega$.  The detector at left endpoint is controlled by Alice, and the detector at the right endpoint is controlled by Bob.  Alice and Bob can change the position of their detectors independently of one another. 
Let $(x^\mu) = (t,s)$ denote a coordinate frame on physical spacetime with respect to which Alice and Bob are both stationary, and let $\Sigma_t$ be the foliation of spacetime by constant $t$-slices.
According to the Hypersurface Bohm-Dirac theory \cite{DGMZ99}, with respect to $\Sigma_t$ the guiding equation for each of the two particles is
\begin{equation}
    \frac{dQ_1}{dt} = \frac{ j^{10}}{j^{00}}(t,Q_1(t),t,Q_2(t)),\qquad \frac{dQ_2}{dt} = \frac{ j^{01} }{j^{00}}(t,Q_1(t),t,Q_2(t)).
\end{equation}
 Suppose the source emits a pair of particles at $t=0$. Given a trajectory of the two-body system $(Q_1(t),Q_2(t))$, its first arrival time chez Alice is the earliest time that either of the two particles reaches Alice's detector, and similarly for Bob.  
 In other words
\begin{equation}
    T_A := \min_{k=1,2} \inf\{t>0\ | \  Q_k(t)\leq -L\}, \qquad  T_B := \min_{k=1,2} \inf\{t>0\ |\  Q_k(t) \geq L\}.
\end{equation}
$T_A$ and $T_B$ are random variables on $\Omega$ (due to their dependence on the initial positions of the trajectories) describing the first of the two arrival times at Alice's detector, resp. Bob's.  Let $\mu_A$ and $\mu_B$ denote their corresponding (un-normalized) probability densities.  We have
\begin{prop}
    \begin{eqnarray}       
    \label{alice_dist}
        \mu_A(t) & = & \int_{-L}^L\left(\sum_{\sigma_2} |\psi_{+\,\sigma_2}(t,-L,t,s')|^2 + \sum_{\sigma_1}|\psi_{\sigma_1\,+}(t,s',t,-L)|^2\right)\ ds'\nonumber \\ && \mbox{}+ \int_0^t dt' \left(|\psi_{-+}(t',L,t,-L)|^2 + |\psi_{+-}(t,-L,t',L)|^2\right) \\ 
    \label{bob_dist}
        \mu_B(t) & = & \int_{-L}^L\left(\sum_{\sigma_2} |\psi_{-\,\sigma_2}(t,L,t,s')|^2 + \sum_{\sigma_1}|\psi_{\sigma_1\,-}(t,s',t,L)|^2\right)\ ds'\nonumber\\ && \mbox{} + \int_0^t dt' \left(|\psi_{-+}(t,L,t',-L)|^2 + |\psi_{+-}(t',-L,t,L)|^2\right).
    \end{eqnarray}
\end{prop}
\begin{proof}

Let $T^f$ be the first time (according to $\Sigma_t$ foliation) at which a particle is registered by {\em any} of the
detectors at the boundary, let $Z^f = (T^f,S^f)$ with $S^f \in \{ L,-L \}$ be the corresponding detection event, and let $I^f \in \{1,2\}$ be 
the label of the registered particle. The proposed rule in \cite{Tumulka}, specialized to this particularly simple case, asserts that the joint probability distribution of $I^f$ and $Z^f$ is
\begin{eqnarray}
    \mbox{Prob}\left(I^f = 1, (t<T^f< t + dt, S^f = \pm L)\right) & = & ds\,dt\ \delta(s\mp L)\int_{-L}^L  \pm j^{10}(t,s,t,s_2) ds_2 \\ 
    \mbox{Prob}\left(I^f = 2, (t<T^f<t+ dt, S^f = \pm L)\right) & = & ds\,dt\ \delta(s\mp L)\int_{-L}^L  \pm j^{01}(t,s_1,t,s) ds_1.
\end{eqnarray}
Since these events are disjoint, it follows that the probability density for the arrival time of either particle at Alice {\em before the other one arrives at Bob} is
\begin{equation}
   \mu_A^f(t) = \int_{-L}^L - j^{10}(t,-L,t,s')-j^{01}(t,s',t,-L) \ ds,
\end{equation}
and similarly for Bob.  We note that when the boundary conditions \eqref{2boundary} are satisfied, the above becomes
\begin{equation}
 \mu_A^f(t) =  \int_{-L}^L\sum_{\sigma_2} |\psi_{+\,\sigma_2}(t,-L,t,s')|^2 + \sum_{\sigma_1}|\psi_{\sigma_1\,+}(t,s',t,-L)|^2\ ds' \geq 0.
\end{equation}

To account for all the particle arrivals at a given detector, we need to include the possibility of the particle arriving at that detector being the {\em second} one to arrive at any detector, e.g. one of the particles arrives at Bob's detector at a time $t'<t$ before the other one arrives at Alice's at time $t$.  In such a case, the Absorbing Boundary Rule proposed in \cite{Tumulka,Tum22several} stipulates that the detection of the particle by Bob causes the two-body wave function to undergo collapse, and the evolution of the remaining particles will henceforth be governed by the collapsed (i.e. the {\em conditional}) wave function. Let $\phi^{jB}(t',\cdot)$ denote the conditional wave function that results from Bob detecting particle labeled $j \in \{1,2\}$ at time $t'$.  Then
\begin{equation}
    \phi_{\sigma_2}^{1B}(t',s_2) = \frac{\psi_{-\,\sigma_2}(t',L,t',s_2)}{\sqrt{\int_{-L}^L ds' \sum_{\sigma_2} |\psi_{-\,\sigma_2}(t',L,t',s')|^2}},\quad \phi_{\sigma_1}^{2B}(t',s_1) = \frac{\psi_{\sigma_1\,-}(t',s_1,t',L)}{\sqrt{\int_{-L}^L ds' \sum_{\sigma_1} |\psi_{\sigma_1\, -}(t',s_1,t',L)|^2}}.
\end{equation}
(Note that the leftover spin components corresponding to the absorbed particle are zero thanks to the boundary conditions \eqref{2boundary}.) 

Let $T_A^j$ be the time of arrival of particle labeled $j$ at Alice's detector, and let $\hat{j}$ denote the label of the other particle, i.e. $\hat{1}=2$ and $\hat{2} = 1$.  By the formula for conditional probabilities, we have
\begin{eqnarray*}
    \mbox{Prob}(t<T_A& < & t+dt) = \sum_{j=1}^2\left(\mbox{Prob}(t<T_A^j<t+dt, I^f = j) + \mbox{Prob}(t<T_A^j<t+dt, I^f = \hat{j})\right) \\
    & = & \mu_A^f(t) + \sum_{j=1}^2 \int_0^t dt' \mbox{Density}(t'<T_B^{\hat{j}} < t' + dt' , I^f = \hat{j}) \mbox{Prob}(t<T_A^j<t+dt\ |\ T_B^{\hat{j}} = t', I^f = \hat{j}).
\end{eqnarray*}
To calculate the conditional probability that shows up on the last line, we need to use the conditional wave function, viz. for $j=2$,
$$
\mbox{Prob}(t<T_A^2<t+dt\ |\ T_B^{1} = t', I^f = 1) = \sum_{j=1}^{2} |\phi_-^{1B}(t,-L)|^2+|\phi_+^{1B}(t,-L)|^2
= \frac{|\psi_{-+}(t',L,t,-L)|^2}{\sum_{\sigma_2} \int_{-L}^L ds' |\psi_{-\,\sigma_2}(t',L,t',s')|^2},
$$
and similarly for $j=1$, while
\begin{equation}
  \mbox{Density}(t'<T_B^{1} < t' + dt' , I^f = 1) = \sum_{\sigma_2} \int_{-L}^L ds' |\psi_{-\,\sigma_2}(t',L,t',s')|^2,
\end{equation}
which cancels the denominator in the above, establishing the claim.
\end{proof}
\section{Detection vs. Arrival Time}
How well does the ABCD method capture the detection phenomena, that is to say, the interaction of the particles emitted by the source inside the domain with the (presumably macroscopic) detectors  places on the boundary of the domain?  If we imagine that the detector at the left endpoint of the interval $\Omega$ is controlled by Alice, and the one on the right by Bob, it should for example be the case that
 the distribution of arrival times at Alice is not immediately affected if Bob decides to move his detector by an appreciable amount, or if he simply switches his detector off (meaning the boundary condition is not imposed on Bob's side.) 
 For a single particle, this is easily obtained (see below) but for two particles it is not at all obvious, since moving one of the detectors or switching it off changes the two-body wave function, and therefore the individual trajectories of an entangled pair of particles change as soon as one of them enters the domain of influence of Bob's boundary point (since the velocity of each particle depends on the actual positions of both particles).  Nevertheless, one expects the distribution of arrival times not to change faster than what is allowed by relativity.
\begin{prop}
For a single particle, and for all $\epsilon \in (0,L/2)$, the distribution of arrival times at $s=-L$ is unaffected for $t<2L-\epsilon$ when the detector at $s=L$ is shifted to $s=L+\epsilon$.
\end{prop}
\begin{proof}
By \eqref{mu}, the density of arrival times at the boundary point $s=-L$ is
$
   \rho_A(t) =  - j^{1}(t,-L) = |\psi_{+}(t,-L)|^{2}
$.

We next note that due to the translation invariance of the equations, the situation with Bob's detector shifted to $L+\epsilon$ while keeping the particle source and Alice's detector fixed, is equivalent to shifting Bob's detector $\epsilon/2$ units to the right while shifting both the source and Alice's detector $\epsilon/2$ units to the left. This observation allows us to use the solution formulas developed in the above for $\Omega$ that was symmetric with respect to the origin, also to the case where that symmetry is broken due to the shift in Bob's detector.

After such a shift, therefore, the interval is $\Omega' = [-L',L']$, where $L' = L+\epsilon/2$, and the initial data is $\mathring{\psi}'_\pm(s) = \mathring{\psi}_\pm(s+\frac{\epsilon}{2})$. If we denote by $\psi'$ and $\Phi'$ the solutions to 
 \eqref{ibvp_arrival} and \eqref{kg} after the shift, 
from \eqref{chi0} we would then have that
\begin{equation}
\label{eq:ker}  
\psi'_+(t,s) = \Phi'_+(t,s) - H(t-L'+ s)X(t,s) + i H(t-L'- s) Y(t,s),
\end{equation}
where $X$ and $Y$ are certain integral operators acting on the initial data $\mathring{\psi}'$.  In particular,
$$
Y(t,s) = \sum_{\sigma\in\{+,-\}}\int_{L'+s}^t d\xi\int_{-L'-t+\xi}^{-L'+t-\xi}d\zeta \ K_\sigma(t,\xi,L'+s,L'+\zeta) \mathring{\psi}'_\sigma(\zeta)
$$
for certain explicit Kernels $K_\pm$.
Assuming $0<t<2L$, upon evaluating at $s=-L'$ the first Heaviside function in \eqref{eq:ker} is zero, therefore we are left with
$$
\psi_+'(t,-L') = \Phi'_+(t,-L') + i Y(t,-L')
$$
and we can check that changing the inner variable of integration to $\zeta'=\zeta+\frac{\epsilon}{2}$ we obtain
$$
Y(t,-L') = \sum_\sigma\int_{0}^t d\xi\int_{-L'-t+\xi}^{-L'+t-\xi}d\zeta K_\sigma(t,\xi,0,L'+\zeta) \mathring{\psi}'_\sigma(\zeta)
=\sum_\sigma \int_{0}^t d\xi\int_{-L-t+\xi}^{-L+t-\xi}d\zeta' K_\sigma(t,\xi,0,L+\zeta') \mathring{\psi}_\sigma(\zeta')
$$
which is manifestly independent of $\epsilon$. From the formula \eqref{sol_phi_byparts} one can similarly conclude that $\Phi'_+(t,-L')$ is independent of $\epsilon$ (or alternatively, one can use the domain of dependence property that holds for solutions to \eqref{kg}.)  It thus follows that $\psi'_+(t,-L')$ is independent of $\epsilon$, and therefore so is $\rho_A(t)$.
\end{proof}
As a simple corollary of the above, we note that if for the 2-body problem we start with a pure product initial state $\mathring{\psi}^1 \otimes \mathring{\psi}^2$, the solution would remain a pure product as well (the equations \eqref{multiDir} are for the non-interacting case), and therefore the same reasoning as above applies to show that the distribution of first arrival times at Alice will not be affected by a change in the position of Bob's detector.  

Next we can try to answer the same question for two particles whose wavefunction is in
a {\em maximally entangled state}, i.e., 
the solution $\psi$ to the two-body problem \eqref{multiDir} with boundary conditions \eqref{2boundary} whose initial data is
 \begin{equation}\label{2bodyinit}
     \psi(0,s_1,0,s_2) = \frac{1}{\sqrt{2}} \left( \mathring{\psi}^{1}\otimes\mathring{\psi}^{2} + \mathring{\psi}^{3}\otimes\mathring{\psi}^{4} \right),
 \end{equation}
 where $\mathring{\psi}^i \in C^0_c(\Omega)$ for $i=1,\dots,4$
 are four normalized mutually orthogonal 1-body wavefunctions. 
   Without much loss of generality, we can take these to be
    \begin{equation}\label{initial_data}
                \left\{
                \begin{array}{rcl}
                    \mathring{\psi}_{-}^1 & \equiv & 0 \\
                    \mathring{\psi}_{+}^1(s) & = & f_{\mu,\alpha}(s),
                \end{array}\right.  
                \left\{
                \begin{array}{rcl}
                    \mathring{\psi}_{-}^2(s) & = &f_{\mu,\alpha}(s) \\
                    \mathring{\psi}_{+}^2 & \equiv & 0,
                \end{array}\right.  
                \left\{
                \begin{array}{rcl}
                    \mathring{\psi}_{-}^3 & \equiv & 0 \\
                    \mathring{\psi}_{+}^3(s) & = &g_{\mu,\alpha}(s),
                \end{array} \right.  
                \left\{
                \begin{array}{rcl}
                    \mathring{\psi}_{-}^4(s) & = & g_{\mu,\alpha}(s) \\
                    \mathring{\psi}_{+}^4 & \equiv & 0.
                \end{array} \right.
    \end{equation}
    where $f_{\mu,\alpha}$ and $g_{\mu,\alpha}$ are  continuous functions compactly supported in $(\mu-\alpha,\mu+\alpha)$, with the property that $\int f^2 ds = \int g^2 ds = 1$ and $\int fg ds = 0$, and 
    we can assume
    $$ f_{\mu,\alpha}(s) = f_{0,\alpha}(s-\mu),\qquad g_{\mu,\alpha}(s) = g_{0,\alpha}(s-\mu). $$
    

For $i=1,\dots,4$, let $\psi^i$ be the solution to the one-body problem \eqref{rodi_eq} with initial data $\mathring{\psi}^i$.  Clearly, the solution to the 2-body system (\ref{multiDir}--\ref{2boundary}) is 
$\psi = \frac{1}{\sqrt{2}} \left( {\psi}^{1}\otimes {\psi}^{2} +{\psi}^{3}\otimes {\psi}^{4} \right)$.  

Recall that for $t < 2L$ the only non-zero term in the series \eqref{chipm} is the one with $n=0$.   Noticing that $\psi^i(t,-L) = 0$ due to the boundary condition, from \eqref{alice_dist} the distribution of first arrival times at Alice when Bob's detector is shifted to the right by the amount $\epsilon$, with $\alpha << \epsilon < L/2$ and letting $L' := L+\frac{\epsilon}{2}$, becomes
    \begin{eqnarray}  
    \label{step_3}
  \!\!\!\!  \mu_A^\epsilon(t) & := & \frac{1}{2}\int_{-L'}^{L'} ds 
    \Bigg(\Big|\psi_+^1(t,-L')\psi_-^2(t,s)+\psi_+^3(t,-L')\psi_-^4(t,s)\Big|^{2} +\Big|\psi_+^1(t,-L')\psi_+^2(t,s)+\psi_+^3(t,-L')\psi_+^4(t,s)\Big|^{2}\nonumber \\
    & & \mbox{\hspace{.25in}} +\Big|\psi_+^1(t,s)\psi_+^2(t,-L')+\psi_+^3(t,s)\psi_+^4(t,-L')\Big|^{2} + \left| \psi_-^1(t,s)\psi_+^2(t,-L')+\psi_-^3(t,s)\psi_+^4(t,-L')\right|^2\Bigg) \\
 & + &\!\!\!\! \frac{1}{2}\int_0^t dt' \Bigg( \Big|   \psi_-^1(t',L')\psi_+^2(t,-L') + \psi_-^3(t',L')  \psi_+^4(t,-L')   \Big|^2 + \Big| \psi_+^1(t,-L')\psi_-^2(t',L') + \psi_+^3(t,-L')  \psi_-^4(t',L')       \Big|^2\Bigg).\nonumber
    \end{eqnarray}
     In the above, $\psi^j_\pm$ are computed using initial data profiles $f_{-\epsilon/2,\alpha}$ and $g_{-\epsilon/2,\alpha}$. 
The goal would be to show that for $t<2L-\epsilon$ the above expression does not depend on $\epsilon$.  

It is fairly easy to verify this when both particles are {\em massless}, i.e. if we set $m=0$.  This is because in that case the initial data is just transported along characteristics, so that the boundary condition has no effect, and the integrals in \eqref{step_3} are easily seen to be independent of $\epsilon$.  

For massive particles, though, a more detailed analysis of the terms in \eqref{step_3} become necessary.  We plan to carry this out in future work.

\section{Summary and Outlook}
In this paper we have demonstrated that there exists an explicit, unique solution to the initial-boundary-value problem for the 1-body Dirac equation in an interval in one-dimensional space, with absorbing boundary conditions as proposed in \cite{Tumulka} modeling the presence of a detector at each endpoint, and that the solution inherits the regularity of the initial data and depends continuously on it.

We then studied the two-body problem corresponding to a pair of particles emanating from a source inside the interval, and obtained explicit formulas for the probability density of first arrival times at detectors situated at each end point of the interval. We posed the question whether it is possible to affect the distribution of arrival times at one detector by shifting the other detector, and to do so faster than allowed by relativity. We showed that  this is not possible if the initial data is a pure product state, of if the particles are massless. In future work we plan to show that even for an entangled pair of massive particles this is still not possible, so that superluminal signaling is generally ruled out in this model.



\section{Acknowledgements}
We would like to thank Shelly Goldstein, Roderich Tumulka, Markus N\"oth, and Siddhant Das for helpful discussions. We  especially thank Roderich Tumulka for pointing out an error in an earlier draft of our paper, and for suggesting to us the content of Proposition 1 and the outline of its proof.
\bibliographystyle{plain}

\end{document}